\documentclass[twocolumn]{autart}    

\usepackage{graphicx}          

\usepackage{latexsym}
\usepackage{amsmath}
\usepackage{amsfonts}
\usepackage{epstopdf}
\usepackage{amssymb}
\usepackage{bm}

\usepackage{ntheorem}
\usepackage{algorithm}
\usepackage[noend]{algpseudocode}

\usepackage{float}
\usepackage{soul}
\usepackage{xcolor}
\usepackage{comment}

\newtheorem{theorem}{Theorem}
\newtheorem{lemma}[theorem]{Lemma}
\newtheorem{proposition}[theorem]{Proposition}

\newtheorem{remark}[theorem]{Remark }
\newtheorem{example}[theorem]{Example}
\newtheorem{definition}[theorem]{Definition}

\usepackage{color}
 
\usepackage{xcolor}
\usepackage[center]{caption}
\definecolor{verde}{rgb}{0,0.7,0}

\usepackage[colorlinks=true,linkcolor=blue,urlcolor=blue,citecolor=blue]{hyperref}

\newcommand{\be}{\begin{equation}}
\newcommand{\ee}{\end{equation}}

\newcommand{\bea}{\begin{eqnarray}}
\newcommand{\eea}{\end{eqnarray}}
\newcommand{\bean}{\begin{eqnarray*}}
\newcommand{\eean}{\end{eqnarray*}}

\begin{document}

\begin{frontmatter}

\title{\LARGE \bf
On the Herdability  of Linear Time-Invariant Systems with Special  
Topological Structures} 

\thanks[footnoteinfo]{Corresponding Author. Provisionally accepted in Automatica, currenty under review.
}

\author[Padova]{Giulia De Pasquale}\ead{giulia.depasquale@phd.unipd.it},    
\author[Padova]{Maria Elena Valcher\thanksref{footnoteinfo}}\ead{meme@dei.unipd.it}              

\address[Padova]{ Dipartimento di Ingegneria dell'Informazione
 Universit\`a di Padova, 
    via Gradenigo 6B, 35131 Padova, Italy,}  


\begin{abstract}                          
In this paper we investigate the herdability property, namely the capability of a system to be driven towards 
the (interior of the) positive orthant, for   linear time-invariant state space models. 
 Herdability 
of certain matrix pairs $(A,B)$, where $A$ is the adjacency matrix of a   multi-agent network, 
and $B$ a selection matrix that singles out a subset of the agents (the ``network leaders"),  is explored.
The cases when   the   graph  associated with $A$, ${\mathcal G}(A)$, is directed and clustering balanced   (in particular, structurally balanced), or 
it  has a   tree topology and there is a single leader, are investigated.
\end{abstract}

  \end{frontmatter}
\section{Introduction}

 Networked multi-agent systems   have been the subject of an impressive number of contributions in the last two decades,  due to their wide range of application \cite{antsaklis2007special,baillieul2007control,OlfatiFaxMurray,zhang2012network}.   As a result,  the  controllability   of this class of systems, namely the property of the system state to be driven towards any point of the state space, has 
 attracted a lot of interest, mainly aimed at deriving conditions that rely on the communication graph structure, rather than on the specific weights attributed to the graph edges \cite{egerstedt2012interacting,struct_controllability,ParlangeliNotarstefano,graph_controllability,sign_controllability,YaziciogluAbbasEger,ZhangCaoCamlibel}. 
However, there are many research fields, such as biology \cite{Jacquez}, chemistry \cite{comp_gen}, sociology \cite{sociology1}, neuroscience \cite{gupta},   social networks \cite{assign_and_appraise,HomophilyMei} 
etc. for which, due to the nature of the   describing variables, investigating if the system state can be brought towards any point of the state space   is not of practical interest, and may lead to  overly restrictive conditions on the model into play.   Consider, for example, the model of a chemical  reactor and   assume that the state vector  represents some  reactant concentrations. In this context,   it is pointless to impose that the state entries may assume any real value, including the negative ones, while it makes sense to impose that the concentrations of all the elements in the chemical reactor can be brought over a minimum level.  
 When dealing with ecological systems, describing the coexistence of different species in the same habitat,
it is of interest to maintain all population levels above  specific thresholds in order to prevent their extinction.  In the context of   marketing advertisement, it is of interest to devise strategies  targeting    some individuals to bring the consumption level/usage of a certain good/service  for  a  group of consumers over a certain threshold.  In many electoral systems there is an   election threshold  that represents  the minimum share of   votes which a candidate or political party has to achieve to become entitled to any representation in a legislature. 
 It is in  contexts like these, in which (positive) thresholds come into play, that the investigation of a weaker concept with respect to controllability, known in the  literature as herdability \cite{Ruf_Shamma,Ruf_arxiv}, becomes of interest.
Herdability refers to the possibility of driving the   state variable towards the interior of the positive orthant.  More precisely, a system is herdable if, for every choice of the initial conditions, there exists a control input that drives all the state variables over a positive threshold.  Clearly, controllable systems are also herdable, but the converse is not true.

  While there is an extensive literature on controllability and structural controllability of networked systems \cite{egerstedt2012interacting,struct_controllability,ParlangeliNotarstefano,graph_controllability,YaziciogluAbbasEger,ZhangCaoCamlibel},
the research  on herdability is still at an early stage.   In \cite{Meng_herd} the herdability of leader/follower networks is studied. The leaders are assumed to be equipped with an external control and the relationships among the agents in the network can be either cooperative or competitive. The leader group selection problem is investigated in this context, with a special focus on structurally balanced network configurations.  In  \cite{She_herd}  the herdability of leader/follower signed networks is investigated starting from the network topology and  some sufficient conditions for herdability based on $1$-walks and $2$-walks in the graph are provided.  The results are then extended to the case of acyclic graphs  with   walks of arbitrary length.
In \cite{Ruf_arxiv} a connection between system herdability and sign distribution over some specific graph topologies is established.  The notion of \textit{sign herdability} is introduced in association with classes of systems whose herdability is deduced from their sign patterns. In  \cite{She_automatica}, a characterization of the controllable subspace is given based on (generalised equitable)  graph partitions, and some sufficient conditions for the system herdability are provided. The concept of quotient graph is also exploited in the study of herdability of the original graph. Finally, in \cite{Ruf_Shamma}   how the underlying graph structure affects the system herdability is investigated for signed and directed graphs, and the herdability of a subset of nodes in a graph is studied, also focusing on the herdability of  directed out-branching rooted graphs with a single input. 

  In this paper we study the herdability of linear time invariant systems described by a matrix pair $(A,B)$,  assuming that the matrix $A$ (and the associated graph ${\mathcal G}(A)$)   represents a multi-agent system and in particular a social network, while $B$ is a selection matrix that singles out the agents that are subject to a direct control action (the ``leaders"). We focus on special topologies of the graph $\mathcal{G}(A)$,   as the tree topology,   or the case when the graph is directed and structurally/clustering balanced, configurations that are quite typical in social networks.
    

 
In detail, in Section \ref{s2} some preliminaries,  together with the definition of  herdability of a matrix pair $(A,B)$ and the main characterisation available in the literature, are given. 
Section \ref{s4} investigates  herdability of networked systems with leader/follower topologies and a directed communication graph that is structurally or clustering balanced. On the other hand,    Section \ref{s5} focuses on the case of tree topologies and a single leader. Section \ref{s6} brings some conclusions, while the Appendix provides all the technical results required to prove the results  of Sections \ref{s4} and \ref{s5}.

 This paper extends  the preliminary results included in our recent paper \cite{ge_herdcdc21}. Specifically,    
  Propositions \ref{p1} and \ref{p2} in section \ref{s5} 
can also be found in \cite{ge_herdcdc21}, while Proposition \ref{p3} has appeared in \cite{ge_herdcdc21} without a proof.
Lemma \ref{lemmaB} extends Lemma 3 in 
\cite{ge_herdcdc21}, while Proposition \ref{ablocchi_new} provides a generalisation of Proposition 7 in \cite{ge_herdcdc21}. All the remaining results are new.

\section{Preliminaries and   definition of herdability of a pair $(A,B)$}\label{s2}

We start the paper by providing some basic definitions and notation that will be used in the following.
\\
Given   $k, n\in \mathbb{Z}$, with $k <n$,   the symbol   $[k,n]$   denotes the  integer set  $\{k, k+1, \dots, n\}$.
The $(i,j)$-th entry of a matrix $A$ is denoted 
 by $[A]_{ij}$, while the $i$-th entry of a vector ${\bf v}$ by $[{\bf v}]_i$.
  The notation $M= {\rm diag} \{m_1, m_2,  \dots,  m_n\}$ indicates a   diagonal matrix  with diagonal entries $m_1, m_2,  \dots,  m_n$.
We let   ${\bf e}_i$ denote the $i$-th vector of the canonical basis of $\mathbb{R}^n$,  where the dimension $n$ will be   clear from the context. 
 Accordingly, $M {\bf e}_j$ denotes the $j$-th column of $M$,  and
 ${\bf e}_i^\top M$ the $i$-th row of $M$.
  Every nonzero multiple of a canonical vector is called {\em monomial vector}.
The vectors ${\bf 1}_n$ and ${\bf 0}_n$ denote the $n$-dimensional vectors whose entries are all $1$ or $0$, respectively.
 Similarly, the symbol ${\bf 0}_{p \times m}$ denotes the $p\times m$ matrix with all zero entries.

   Given a vector ${\bf v} \in {\mathbb R}^n$, the set    
   $\overline{\rm ZP}({\bf v}) = \{ i \in [1,n]: [{\bf v}]_i \neq 0\}$  denotes the  
   {\em non-zero pattern} of ${\bf v}$  \cite{TACconPaolo}.  
   A nonzero vector ${\bf v}$ is said to be {\em unisigned} \cite{Ruf_arxiv} if all its nonzero entries have the same sign.
%
 Given a matrix $A\in \mathbb{R}^{n \times m}$, the notation ${\rm Im (A)}$ denotes the image of the matrix $A$.
 A matrix  (in particular, a vector)
 $A$ is   {\em nonnegative}  (denoted by  $A \ge 0$) \cite{BookFarina}  if all its entries are nonnegative.
  $A$ is   {\em strictly positive} (denoted by  $A \gg 0$) if all its entries   are positive.
  A matrix $P\in {\mathbb R}^{n \times n}$ is a {\em permutation matrix} if its columns are a permuted version of the columns of the identity matrix $I_n$.
To any matrix   $A\in {\mathbb R}^{n \times n}$, we associate the {\em signed and weighted directed graph}
${\mathcal G}(A)= ({\mathcal V}, {\mathcal E}, A),$
where ${\mathcal V} =[1,n]$ is the set of nodes.
The set ${\mathcal E}\subseteq {\mathcal V}\times {\mathcal V}$ is  the set of arcs  (edges)
connecting  the nodes, while the matrix $A\in {\mathbb R}^{n \times n}$ is the   {\em adjacency matrix} of the graph. There  is an arc $(j, i) \in \mathcal{E}$ from $j$ to $i$ if and only if $[A]_{ij} \ne 0$.  When so, $[A]_{ij}$ is the {\em weight} of the arc.
  A sequence of $k$ consecutive arcs 
 $(j, j_2),(j_2, j_3),$ $\dots, (j_k,i)\in {\mathcal E}$ is a {\em walk}  
 of length $k$ 
 from $j$  to 
 $i$.
A walk from $j$   to $i$ is said to be {\em positive (negative)} if the product of the weights of the edges that compose the walk is positive (negative). 
A {\em minimum walk} from $j$   to $i$ is a walk of minimum length connecting the two nodes. We define the   {\em distance} $d(j,i)$ from the node $j$ to the node $i$
as the length of the minimum walk from $j$ to $i$. If there is no  walk from $j$ to $i$ then $d(j,i)=+\infty$. 
 If $A$ is a symmetric matrix, namely $A=A^\top$, the graph ${\mathcal G}(A)$ is (signed, weighted and) {\em undirected}, and 
 the concepts of walk and distance become symmetric.
  An undirected graph ${\mathcal G}(A)$ is {\em  connected} if for every pair of vertices 
  there is a walk connecting them.

A graph ${\mathcal G}(A)$ is said 
to be {\em clustering balanced} (with $k\ge 2$ clusters)  \cite{Bullo2020,davis67} if all    its nodes can be partitioned into $k$ disjoint subsets ${\mathcal V}_1, {\mathcal V}_2, \dots, {\mathcal V}_k$   in such a way that $\forall i,j \in    {\mathcal V}_p$, $p \in [1,k]$, we have $[A]_{ij} \geq 0$, and $\forall i \in {\mathcal V}_p$ and $\forall j \in {\mathcal V}_q$, $p,q \in  [1,k],$ $p \neq q$,  we have $[A]_{ij}\leq 0$. 
Note that clustering balance for $k=2$ clusters is generally known in the literature as {\em structural balance} 
\cite{Altafini2013,Bullo2020,davis67}. 
\medskip

%


The concept of herdability of linear and time-invariant state space models
 described by a matrix pair $(A,B)$, with $A \in {\mathbb R}^{n \times n}$ and $B \in {\mathbb R}^{n \times m}$,
 has been defined in various ways \cite{Ruf_Shamma,Ruf_arxiv,She_herd}. 
In this paper we are   interested 
 in the behavior of all state variables, rather than in the behavior of a subset of them.   Consequently, 
 we assume the following   definition
 (which is equivalent to Definition 3 in \cite{Ruf_arxiv}). 
 \smallskip

\begin{definition}  Given a (continuous-time or discrete-time)   (linear and time-invariant) state space model of dimension $n$ with $m$ inputs,
described by a pair $(A,B), \ A\in {\mathbb R}^{n\times n}$ and $B\in {\mathbb R}^{n\times m}$, 
the system (the pair) is said to be {\em 
herdable} if for every ${\bf x}(0)$ and every $h> 0$, there exists a time $t_f > 0$ and an input ${\bf u}(t), t\in [0,t_f),$ that drives the state of the system from ${\bf x}(0)$ to ${\bf x}(t_f) \ge h {\bf 1}_n$\footnote{
It is worth noticing that, when dealing with linear systems, requiring that each of the agents' states can be brought above a common positive threshold $h$ is equivalent to requiring that each state variable $x_i$  can be brought above a specific positive threshold $h_i$ or that at some time instant the state vector becomes strictly positive.}.
\end{definition}

Both in the continuous-time case and in the discrete-time case, herdability reduces to a condition on the controllability matrix associated with the pair $(A,B)$.
\smallskip

 \begin{proposition} [Corollary 1, \cite{Ruf_arxiv}] A pair $(A,B), A\in {\mathbb R}^{n\times n}$ and $B\in {\mathbb R}^{n\times m}$, is herdable if and only if 
  ${\rm Im}({\mathcal R}(A,B))$ includes   a strictly positive vector, where 
\be
{\mathcal R}(A,B) := \begin{bmatrix} B & AB & A^2B & \dots & A^{n-1} B \end{bmatrix}
\label{reach_mat}
\ee
is the {\em controllability matrix} of the pair $(A,B)$.
\end{proposition}


%

\section{Herdability of pairs $(A,B)$ corresponding to a directed graph ${\mathcal G}(A)$ with $m$ leaders}\label{s4}

In this section we
investigate the herdability  of
the pairs $(A,B)$, 
where $A\in {\mathbb R}^{n\times n}$ is any real matrix and $B \in {\mathbb R}^{n\times m}, n > m, $ is a {\em selection matrix}, namely an $n \times m$ submatrix of the identity matrix $I_n$. 
This set-up, previously considered in \cite{Meng_herd,She_herd,She_automatica}, can be interpreted as the description  of a network of mutually interacting agents, each of them associated with a scalar describing variable.
In the network, a subset   of $m$ agents (the indices of the nonzero rows of $B$) is selected as ``leaders", by this meaning that such agents are the target of a direct external action, aiming to influence the states  of the remaining agents, the  ``followers"\footnote{\color{blue}Note that it is the choice of where to apply the external inputs, and not the communication network, that determines the leaders.}.
This set-up is also very similar to the one adopted in
\cite{YaziciogluAbbasEger,ZhangCaoCamlibel},  where controllability properties of continuous systems, known in the literature as ``networks of diffusively coupled agents" and described by a pair
$\left(- {\mathcal L}, B \right)$, where ${\mathcal L}$ is a Laplacian matrix and $B$ a selection matrix, 
have been  investigated.
Both in   references \cite{YaziciogluAbbasEger,ZhangCaoCamlibel} and in 
\cite{Meng_herd,She_herd,She_automatica}, the key idea is to exploit the structure  of the signed and weighted directed graph ${\mathcal G}(A)$  and the specific selection of the leaders to deduce some (in general, sufficient) conditions for controllability/herdability  to hold.
Indeed, herdability of a system described by such a pair $(A,B)$   represents the capability of the group of $m$ leaders to 
simultaneously bring their own states and those of the followers over a certain threshold.
The results we will derive in this setting can be used not only to analyse an existing network, but also for design purposes. Indeed, one can look for all possible 
 leaders' selections  that allow this property to hold and, in particular, for the  smallest sets of leaders that ensure herdability of the resulting network.
 The interest in this kind of problems is quite immediate if we think, for instance, of the marketing or the electoral system examples, where 
 it is quite relevant to understand which choices of the leaders ensure the success of the marketing/propaganda policies. Similarly, in an ecological system, 
 identifying policies that target specific individuals in the various populations to prevent their extinction is of utmost importance.

Finally, based on 
Lemma \ref{lemmaTinput} in the Appendix,  
we can always   reduce ourselves to the  case $(A,B)$, with  $B$ a selection matrix,
every time  all the nonzero rows of $B$ are linearly independent.

We  first consider the case when ${\mathcal G}(A)$,  the communication graph associated with $A$, is structurally balanced  or, more generally, clustering balanced  \cite{Davis}. These configurations are of strong interest in sociological contexts since they describe the case
when individuals split in factions: individuals within the same faction have friendly/cooperative behaviors, while individuals belonging to different factions
behave in a competitive/antagonistic way. This may be the case when considering fans supporting different sport teams, political parties supporters or animal species competing for the same natural resources.
In particular, a structurally balanced network
 represents an intrinsically stable social configuration \cite{DeGroot,Heider}.

It must be remarked that bringing all the agents' states over a certain positive threshold in   the  presence of competitive interactions is nontrivial,   especially when  the antagonistic relationships between individuals of different factions tend to stimulate somewhat opposite reactions  that lead to opposite signs of the 
 variables   involved in the system dynamics.
\\
 Proposition \ref{kblocchi}, below, addresses the case when the directed graph ${\mathcal G}(A)$ is clustering balanced,   the set of leaders coincides with one of the clusters (without loss of generality the first one) and for each follower 
 in the other clusters there is a leader whose distance from that follower is smaller than the distance from any other follower belonging to a different cluster.

\smallskip

\begin{proposition}\label{kblocchi} 
Assume that $\mathcal{G}(A)$ is a clustering balanced   directed graph with $k$ clusters, ${\mathcal V}_1, \dots, {\mathcal V}_k$, 
and  that the set of leaders coincides with one of the   clusters, e.g., $\mathcal{L} = {\mathcal V}_1$.
If
 for every $p\in [2,k]$ and  every $i \in {\mathcal V}_p$ there exists $\ell_i\in {\mathcal L}={\mathcal V}_1$ such that  $d(\ell_i, i) < d(\ell_i, j),   \forall j \in \cup_{h\not\in \{1,p\}} {\mathcal V}_h$,
then the pair $(A,B)$ is herdable.
 \end{proposition}

\begin{proof}  
It entails no loss of generality assuming that 
${\mathcal L} = {\mathcal V}_1 = [1,m]$. Therefore
 $B = \begin{bmatrix} I_{m}\cr 0\end{bmatrix}$ and the reachability matrix takes the form
 $${\mathcal R} = \begin{bmatrix} B &\vline & AB & A^2B & \dots & A^{n-1} B \end{bmatrix} =\begin{bmatrix} I_m &\vline& \Phi_{12}\cr 0 &\vline& \Phi_{22}\end{bmatrix}.$$

Under the statement assumptions, for every 
$p\in [2,k]$ and  $\forall i \in {\mathcal V}_p$ there exists $\ell_i \in {\mathcal L}={\mathcal V}_1=[1,m]$ and $k_i> 0$ (the distance from $\ell_i$ to $i$) such that 
$[A^{k_i} B]_{i,\ell_i}\ne 0$ and if 
$[A^{k_i} B]_{j,\ell_i}\ne 0$ for some $j\ne i$ then either $j \in {\mathcal L}$ or $j \in {\mathcal V}_p$.
Clearly, the vector $A^{k_i} B {\bf e}_{\ell_i}$ represents the $(mk_i + \ell_i)$-th column of ${\mathcal R}$, and its restriction to the last $n-m$ entries 
is the $(mk_i - m + \ell_i)$-th column of $\Phi_{22}$.
We want to prove that  such a restriction is a unisigned vector.
To this end, we first observe that assumption $d(\ell_i, i) < d(\ell_i, j),   \forall j \in \cup_{h\not\in \{1,p\}} {\mathcal V}_h$, implies that the walk of length $k_i$ from $\ell_i$ to $i$ cannot pass through any cluster ${\mathcal V}_h, h \not\in \{1, p\},$ therefore the nodes belonging to the walk either belong to ${\mathcal V}_1$ or to ${\mathcal V}_p$.
This immediately implies that if
$[A^{k_i} B]_{j,\ell_i}\ne 0$ for some $j \in {\mathcal L}={\mathcal V}_1$ then $[A^{k_i} B]_{j,\ell_i}> 0$, while if 
$[A^{k_i} B]_{j,\ell_i}\ne 0$ for some $j \in {\mathcal V}_p$ then   $[A^{k_i} B]_{j,\ell_i} <0$.
This implies that the $(mk_i - m + \ell_i)$-th column of $\Phi_{22}$
 is a unisigned vector.
On the other hand, since $\cup_{i\in [2,k]} \overline{\rm ZP}(A^{k_i} B {\bf e}_{\ell_i}) \supseteq [m+1,n]$, this means that ${\rm Im}(\Phi_{22})$ includes a strictly positive vector.
Therefore, by Lemma \ref{lemmaB} in the Appendix (see also Remark \ref{remo_romolo}), ${\rm Im}({\mathcal R})$ includes a strictly positive vector, and hence the pair $(A,B)$ is herdable.
 $\square$
\end{proof}
 \smallskip


Proposition \ref{splittedleaders} considers the case when the directed graph ${\mathcal G}(A)$ is structurally balanced and there are leaders 
in both classes.

 \smallskip
 
\begin{proposition}\label{splittedleaders}    
 Assume that 
the directed graph ${\mathcal G}(A)$ is   structurally balanced, with nodes split into clusters  ${\mathcal V}_1$  and ${\mathcal V}_2$, and that the set of leaders ${\mathcal L}$ intersects both ${\mathcal V}_1$ and ${\mathcal V}_2$.
If  
\begin{itemize}
\item[a)] $\forall i \in {\mathcal V}_1\setminus {\mathcal L}$ there exists $\ell \in {\mathcal L} \cap {\mathcal V}_1$ such that
$d(\ell, i) < d(\ell, j),   \forall j \in {\mathcal V}_2 \setminus {\mathcal L};$
\item[b)] $\forall i \in {\mathcal V}_2\setminus {\mathcal L}$ there exists $\ell \in {\mathcal L} \cap {\mathcal V}_2$ such that
$
d(\ell, i) < d(\ell, j),   \forall j \in {\mathcal V}_1  \setminus {\mathcal L};
$
\end{itemize}
then the pair $(A,B)$ is herdable.
\end{proposition}

\begin{proof} 
 First of all,
we observe that under the structural balance assumption if two nodes (leader or follower) belong to the same class, every     walk (and hence, in particular, every minimum walk) that connects them has a positive weight. As a result, if $i,j \in {\mathcal V}_p$ for some $p\in [1,2]$ and $[A^kB]_{ij}\ne 0$ for some $k >0$, then 
$[A^kB]_{ij} > 0$.\\
Condition a)  ensures that for every $i \in {\mathcal V}_1\setminus {\mathcal L}$ there exists $\ell \in {\mathcal L} \cap {\mathcal V}_1$ and $k_i>0$ such that $[A^{k_i}B]_{i\ell}\ne 0$ and hence $[A^{k_i}B]_{i\ell}> 0$. On the other hand, 
$[A^{k_i}B]_{j\ell}= 0$ for every $j\in {\mathcal V}_2\setminus {\mathcal L}$. Therefore if $[A^{k_i}B]_{j\ell}\ne 0$  and $j\not\in {\mathcal L}$ then 
$[A^{k_i}B]_{j\ell}> 0.$
Consequently, for every $i \in {\mathcal V}_1\setminus {\mathcal L}$ there exists $\ell\in {\mathcal V}_1\setminus {\mathcal L}$ and $k_i>0$ such that 
the vector
$A^{k_i}B{\bf e}_\ell$ has the $i$-th entry which is nonzero and 
its restriction   to the entries that correspond to the followers {\color{blue} (i.e., with indices in $[1,n]\setminus{\mathcal L}$)} is a unisigned vector.
By exploiting b), we can claim the same result for all indices  $ i \in {\mathcal V}_2\setminus {\mathcal L}$. So, keeping in mind the structure of $B$, we can claim 
that there exists a permutation matrix $P$ and a selection matrix $S$ such that
$$P {\mathcal R}(A,B) S = \begin{bmatrix} I_m & \Phi_{12}\cr 0 & \Phi_{22}\end{bmatrix},$$
where all columns of $\Phi_{22}$ are unisigned (in fact, nonnegative) and $\Phi_{22}$ has no zero rows.
By Lemma \ref{lemmaB},   in the Appendix, 
we can claim the herdability of the pair $(A,B)$.
 $\square$\end{proof}
\smallskip

\begin{remark}  
It is worth noticing that conditions a) and b) in Proposition \ref{splittedleaders}  amount to requiring that for each follower there is a
   leader in the same cluster  that is closer to that follower   than to any other follower  belonging to  the other cluster, something reasonable to assume when dealing with social networks.
\end{remark}

\section{Herdability of pairs $(A,B)$ with ${\mathcal G}(A)$ an undirected tree with a single leader}\label{s5}

We now consider the case when $B$ is a canonical vector 
and the matrix $A$ is a symmetric real matrix whose associated undirected graph ${\mathcal G}(A)$ is  
a {\em tree}
\cite{DeoGraph}.
Undirected graphs are reasonably common when describing social networks, due to the fact that
in the long term the friendly/antagonistic attitude that an individual has toward another one tends to be  reciprocated.
Also, it has been shown that moderate size scale social graphs exhibit non-trivial tree-like structures and tree-like decomposition properties \cite{treelike1,treelike2}.
In fact, trees
are typically used  to represent social networks that exhibit  a multi-layer organisation (for instance, employees in a company,  members of a sport association...). 

  When an undirected tree represents a social network topology, it makes sense to assume, as in the previous section, that the leader of the network
is the individual who is   subject to a direct external influence.
All the other $n-1$ members of the network will be referred to as followers.  
This case has been investigated in \cite{She_herd}, where a sufficient condition for the herdability   of such  pairs $(A,B)$ has been provided. 
In this section we provide a  sufficient condition for herdability that is less restrictive, and in the case of trees whose followers have 
distance at most $2$ from the leader we provide necessary and sufficient conditions.

 To investigate the problem we 
adopt 
the following  non restrictive

\noindent {\bf Assumption}: The graph ${\mathcal G}(A)$ is a signed, weighted, connected and acyclic undirected graph, namely a tree \cite{DeoGraph}.
  Let us assume $B={\bf e}_1$, and hence the leader is ${\mathcal L}=\{1\}$,  while the followers   split into classes, based on their   distance from the leader. The followers at distance $1$ from the leader are ${\mathcal F}_1= [2,m_1+1]$, the followers at distance $2$ from the leader are 
${\mathcal F}_2= [m_1 +2,m_1+m_2+ 1]$, and so on till the last class ${\mathcal F}_{k}= [m_1 + \dots + m_{k-1}+2,n]$, where $k$ is the maximum distance between the leader and one of its followers. 
\medskip


\begin{proposition}\label{p1} 
Consider a pair $(A,B)$, with $A\in \mathbb{R}^{n\times n}$ and $B\in {\mathbb R}^n$ satisfying the previous Assumption.\\
If, for every $d\in [0,k-1]$, all the edges from the vertices in ${\mathcal F}_d$ to the vertices in ${\mathcal F}_{d+1}$ have the same sign, 
then the pair $(A,B)$ is herdable.
\end{proposition}

\begin{proof}
Under the previous assumption, it is easy to see that
every vertex in ${\mathcal F}_d$ is reached for the first time  by the leader in $d$ steps,  $d \in [0,k]$, and subsequently it is reached 
after $d+2h$ steps for every $h\in \{1,2,3,\dots\}$ (since each undirected edge of the graph 
can be crossed back and forth, and hence  condition $[A^dB]_{i}\ne 0$ implies $[A^{d+2} B]_{i}\ne 0$). Therefore the controllability matrix of the pair $(A,B)$ takes the form
\begin{equation}\label{herd_tree}
{\mathcal R} 
= \begin{bmatrix}\
1 & 0                 & *             & 0 & * & \dots \cr
0 & {\bf v}_1      & 0            & * & 0 & \dots \cr
0 & 0                 & {\bf v}_2 & 0 & * & \dots \cr
0 & 0                 & 0 & {\bf v}_3 & 0 &   \dots \cr
\vdots &    \vdots              & \vdots &   &   \vdots&   \vdots \cr
0 & 0                 & \dots  &\dots   & {\bf v}_{k}    &  \dots 
\end{bmatrix},
\end{equation}
where ${\bf v}_d\in {\mathbb R}^{m_d}, d\in [1,k],$ are, by assumption, unisigned, while $*$ denotes (nonzero) vectors/entries whose values are not relevant.
So,  by  Remark \ref{remo_romolo} in the Appendix, we immediately deduce that there exists a strictly positive vector in the image of ${\mathcal R}$, and hence $(A,B)$ is herdable.
 $\square$\end{proof}
\smallskip

\begin{remark}
Theorem 3 in \cite{She_herd}  follows   as a corollary of the previous proposition, since it imposes  that all paths from the leader to the followers in 
${\mathcal V}_o := \cup_{h \in {\mathbb Z}_+} {\mathcal F}_{1+2h}$ have the same sign and, at the same time,  
all paths from the leader to the followers in 
${\mathcal V}_e := \cup_{h \in {\mathbb Z}_+} {\mathcal F}_{2+2h}$ have the same sign. This means that not only all the  edges from vertices in ${\mathcal F}_d$ to 
vertices in ${\mathcal F}_{d+1}, d\in [0,k-1]$,  (where ${\mathcal F}_0:={\mathcal L}$) have the same signs, but such signs are uniquely determined for $d\ge 1$ once we choose the signs of the edges from ${\mathcal F}_1$ to  
${\mathcal F}_2$.
\end{remark}
\smallskip

\begin{example} \label{ex1}
Consider a pair $(A,B)$, with  $A=A^\top\in {\mathbb R}^{9\times 9}$ and $B = {\bf e}_1$, and assume that  the   undirected graph ${\mathcal G}(A)$ associated with the matrix $A$ is a tree    whose structure and   edge signs are described in Figure \ref{grafoEx1}.  
\begin{figure}[h!]  
\includegraphics[scale=0.7]{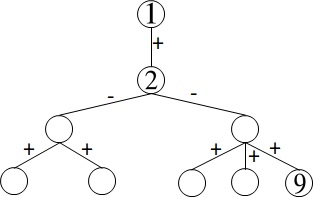}
\centering
\caption{  Tree structure of the  herdable system of Example \ref{ex1}.} \label{grafoEx1}
\end{figure}
 The nodes $i=2$ and $j=9$ both belong to ${\mathcal V}_o$, since both of them are reached from the leader (node $1$ in Fig. \ref{grafoEx1}) in an odd number of steps ($1+2h$ and $3+2h$, $h \in \{0,1,2, \dots\}$, respectively). The node $i$ is reached by the leader with positive walks, while $j$ with negative ones, so the hypotheses of Theorem 3 in \cite{She_herd} are violated. However, the controllability matrix of the pair takes the structure in \eqref{herd_tree} for $k=3$, with unisigned vectors 
 ${\bf v}_1$, ${\bf v}_2$ and ${\bf v}_3$, the first one with a positive entry, while the other two with negative entries,
 thus the pair is herdable by Proposition \ref{p1}.
\end{example}
\smallskip

\begin{remark} As previously remarked, the choice of the leader is not intrinsic 
to the structure of the tree, but is just the specific node to which we   apply the input.
In general, herdability 
is achieved only
when selecting  as leaders certain nodes,   rather than others, as it is related to the  path signs  that connect the leader to the other nodes.
\end{remark}
\smallskip

 Propositions \ref{p2} and \ref{p3}, below, provide complete characterizations of herdability for trees in which followers have all distance $1$ from the leader or distance at most $2$ from the leader, respectively.
\smallskip


\begin{proposition}\label{p2} 
Consider a pair $(A,B)$, with $A\in \mathbb{R}^{n\times n}$ and $B\in {\mathbb R}^n$ satisfying  the aforementioned Assumption, and suppose that all the followers have distance one from the leader.\\
Then the pair $(A,B)$ is herdable if and only if all the edges have the same sign.
\end{proposition}

\begin{proof}
If all the followers have distance $1$ from the leader, namely $k=1$, then
$$A = \begin{bmatrix} 0 & A_{12}\cr A_{21} & {\bf 0}_{(n-1)\times (n-1)}\end{bmatrix},$$
where $A_{21} = A_{12}^\top \in {\mathbb R}^{n-1}$ is  devoid of zero entries.
By   Proposition \ref{ablocchi_new}, $(A,B)$ is herdable if and only if the pair $({\bf 0}_{(n-1)\times (n-1)}, A_{21})$ is herdable, and this is the case if and only if 
$A_{21}$ is unisigned.
 $\square$\end{proof}
\smallskip

\begin{proposition}\label{p3} 
Consider a pair $(A,B)$, with $A\in \mathbb{R}^{n\times n}$ and $B\in {\mathbb R}^n$ satisfying  the aforementioned
  Assumption, and suppose that all the followers have distance at most $2$ from the leader, and hence
$$A = \begin{bmatrix} 0 & A_{12} & {\bf 0}_{1\times m_2}\cr A_{21} & {\bf 0}_{m_1\times m_1} & A_{23}\cr
{\bf 0}_{m_2\times 1} & A_{32} & {\bf 0}_{m_2 \times m_2}\end{bmatrix},$$
where $A_{21} = A_{12}^\top\in {\mathbb R}^{m_1}$ and $A_{32}= A_{23}^\top \in {\mathbb R}^{m_2\times m_1}$. 
 Then the pair $(A,B)$ is herdable if and only if 
for every $i,j\in {\mathcal F}_1= [2, m_1+1]$ (including $i=j$)\footnote{Note that for $i=j$  condition i) becomes trivial, while condition ii) becomes ``$A_{32} {\bf e}_i$ is  either zero or unisigned".} such that 
\be
 [A_{23}A_{32}]_{ii}= 
  [A_{23}A_{32}]_{jj},
\label{noncompatta}
\ee
we have:
\begin{itemize}
\item[i)] $[A_{21}]_i \cdot [A_{21}]_j >0$ (namely 
 the two edges from the leader ${\mathcal L}$ to  $i$ and $j$ have the same sign);
 \item[ii)] $A_{32} ({\bf e}_i + {\bf e}_j)$ is either zero or unisigned (namely all edges from $i$ and $j$ to their followers in ${\mathcal F}_2$ have the same sign).
 \end{itemize}
\end{proposition}

\begin{proof} 
First of all, we highlight that, by Assumption, $\Gamma := A_{21}$ is devoid of zero entries, 
and for every $i\in [1,m_2]$ the row vector ${\bf e}_i^\top A_{32}$ is a monomial vector (namely it has a single nonzero entry). Consequently, 
$\Lambda:= A_{23}A_{32}$ is a diagonal matrix (with nonnegative diagonal entries).

\noindent By   Proposition \ref{ablocchi_new} in the Appendix, $(A,B)$ is herdable if and only if the pair 
$$\left(\begin{bmatrix} {\bf 0}_{m_1\times m_1} & A_{23}\cr
A_{32} & {\bf 0}_{m_2\times m_2}\end{bmatrix}, \begin{bmatrix}A_{21}\cr {\bf 0}_{m_2}\end{bmatrix}\right)$$ 
is herdable, and this is the case if and only if the image of 
the controllability matrix  $\hat {\mathcal R}$ of the previous pair, see \eqref{reachhat},
\begin{figure*}
\be
\hat {\mathcal R} :=
\begin{bmatrix}
A_{21} & 0 & (A_{23} A_{32}) A_{21} & 0 & (A_{23} A_{32})^2 A_{21} & 0 & \dots \cr
0 & A_{32} A_{21} & 0 & A_{32} (A_{23} A_{32}) A_{21} & 0 & A_{32} (A_{23} A_{32})^2 A_{21}  & \dots
\end{bmatrix}
\label{reachhat}
\ee
 \begin{center}
-------------------------------------------------------------------------------------------------------------------------------------------------
\end{center}
\end{figure*}
\noindent includes a strictly positive vector.
This is the case if and only if the following   conditions simultaneously hold:
\begin{itemize}
\item[a)] the image of the controllability matrix
${\mathcal R}_1 :=
\begin{bmatrix}
A_{21} &   (A_{23} A_{32}) A_{21}   & (A_{23} A_{32})^2 A_{21}   & \dots \end{bmatrix}$ includes a strictly positive vector, namely the pair 
$(A_{23} A_{32}, A_{21})$ is herdable;
\item[b)] 
the image of the matrix
 $A_{32}  {\mathcal R}_1$
includes a strictly positive vector.
\end{itemize}
As the matrix $\Lambda = A_{23} A_{32}$ 
 is  diagonal, while the column vector $\Gamma = A_{21}$ has no zero entries,
 by Lemma
\ref{vandermonde}, the pair $(\Lambda,\Gamma)=(A_{23} A_{32}, A_{21})$ is herdable if and only if 
condition
 \eqref{noncompatta}
 implies $[A_{21}]_i \cdot [A_{21}]_j >0$. This means that a) is equivalent to condition i).
\\
 So, we are now remained with proving that if i) (equivalently, a)) holds, then b) and ii) are equivalent.
 If i) holds,  by referring to the proof of Lemma \ref{vandermonde} in the Appendix, we can assume without loss of generality that $\Lambda$ and $\Gamma$ take the form given in \eqref{lambdagamma} and 
 claim that 
 $${\rm Im} \left(A_{32}  {\mathcal R}_1\right) =
 {\rm Im} \left(A_{32} \cdot {\rm diag}
\{{\bm \gamma}_1, \dots,  {\bm \gamma}_p, {\bm \gamma}_{p+1}, \dots, {\bm \gamma}_s\}\right),$$
where ${\bm \gamma}_i\in {\mathbb R}^{n_i}$ is strictly positive if $i\in [1,p]$ and strictly negative if $i\in [p+1,s]$.\\
Set
$W = 
\begin{bmatrix}
 {\bf w}_1 \,| \, \dots \,|\, {\bf w}_p  \,|\, {\bf w}_{p+1} \,|\, \dots  \,|\, {\bf w}_s
\end{bmatrix} := A_{32} \cdot {\rm diag}
\{{\bm \gamma}_1, \dots,  {\bm \gamma}_p, {\bm \gamma}_{p+1}, \dots, {\bm \gamma}_s\},
$
where
each vector ${\bf w}_i$ is obtained 
by combining with the coefficients  of the vector ${\bm \gamma}_i$ (having all the same sign) the columns of $A_{32}$ of indices
$[h_i+1, h_i+n_i]$, where by definition $h_1:=0$, while
 $h_i := n_1 +n_2 +\dots + n_{i-1}$ for $i\in [2,s]$.\\
 We observe that  
 all  columns of $A_{32}$ are either zero (if a vertex in ${\mathcal F}_1= [2, m_1+1]$ has no followers) or have disjoint nonzero patterns, meaning that for every $\ell, m\in [h_i+1, h_i+n_i], \ell\ne m,$
 $\overline{\rm ZP}(A_{32} {\bf e}_\ell) \cap \overline{\rm ZP}(A_{32} {\bf e}_m) =\emptyset$.
 As a result also the columns ${\bf w}_i$ of $W$ are either zero or have disjoint nonzero patterns. \\
We can now conclude that condition b) holds if and only if   ${\rm Im} \left(A_{32}  {\mathcal R}_1\right) =
 {\rm Im} (W)$  contains a strictly positive 
 vector, but this is possible if and only if all vectors ${\bf w}_h$ are unisigned.
 By the way the vectors ${\bf w}_h$ have been obtained, this is possible if and only if condition ii) holds.
 $\square$\end{proof}
\smallskip

Proposition \ref{p3} states that if ${\mathcal G}(A)$ is a tree and the distance from the unique leader is at most $2$, then herdability
is possible if and only if every time the sum of the squares of all arc  weights from a node $i\in {\mathcal F}_1$   to all the nodes in ${\mathcal F}_2$ 
coincides with the sum of the squares of all arc weights from some node $j \in{\mathcal F}_1$  to all the nodes in ${\mathcal F}_2$,
then 
i) the edges from the leader to $i$ and $j$ must have the same sign; ii) all edges from $i$ and $j$ to the nodes 
in ${\mathcal F}_2$  
must have the same sign.
%
\smallskip

\begin{example}  \label{ex2}
Consider the pair $(A,B)$, with
\begin{eqnarray*}
A &=& 
\begin{bmatrix}
0 & A_{12} &0\cr 
A_{21} & 0 & A_{23}\cr
0 & A_{32} &0 \end{bmatrix} =
\begin{bmatrix}
0 &\vline& 1 & a & 2 &\vline& 0 & 0\cr
\hline
1 &\vline& 0 & 0 & 0 &\vline& 0 & 0\cr
a &\vline& 0 & 0 & 0 &\vline& b & c\cr
2 &\vline& 0 & 0 & 0 &\vline& 0 & 0\cr
\hline
0 &\vline& 0 & b & 0 &\vline& 0 & 0\cr
0 &\vline& 0 & c & 0 &\vline& 0 & 0
\end{bmatrix}, \quad
B = {\bf e}_1,
\end{eqnarray*}
 whose graph is given  in Fig. \ref{grafoEx2},  where $a, b$ and $c$ are nonzero real values. Note that ${\mathcal L}=\{1\}, {\mathcal F}_1=[2,4]$ and ${\mathcal F}_2=[5,6]$, so that $m_1=3$ and $m_2=2$.
\begin{figure}[h!]  
\includegraphics[scale=0.7]{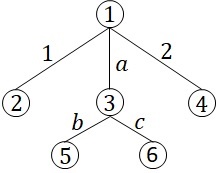}
\centering
\caption{  Graph structure related to the pair $(A,B)$ in Example~\ref{ex2}.} \label{grafoEx2}
\end{figure}
We first check for all indices $i,j\in [1,3], i\ne j,$ whether condition \eqref{noncompatta} holds.
It is easily seen that $[A_{23}A_{32}]_{11}=[A_{23}A_{32}]_{33} = 0.$
In fact, for the pair of indices $(1,3)$ both condition i) 
 and condition ii) of Proposition  \ref{p3} are  satisfied, since $[A_{21}]_1 \cdot [A_{21}]_3 = 2>0$ and both column $1$ and column $3$ of $A_{32}$ are zero.
 On the other hand,  if we assume  $i=j\in [1,3]$, for which   condition \eqref{noncompatta} trivially holds, 
condition i) is straightforward, while condition ii) holds provided that $bc >0$ (namely $b$ and $c$ have the same sign), because all columns of $A_{32}$ are zero
or   unisigned if and only if $bc>0$.
Therefore the pair $(A,B)$ is herdable for every $a\ne 0$ and for $bc>0$. 
\end{example}

\section{Conclusions} \label{s6}

In this paper herdability  of linear time-invariant systems has been investigated.  
First, some 
results for pairs $(A,B)$ 
corresponding to  leader-follower networks ${\mathcal G}(A)$ satisfying clustering balance have been derived.
Subsequently, pairs  $(A,B)$ whose network ${\mathcal G}(A)$ has a   tree topology and a single leader have been addressed.
The study of herdability based on the graph ${\mathcal G}(A)$ seems particularly promising, especially  because  it may lead, for special graph topologies, to determine conditions for ``structural herdability", which is independent of the specific values of the matrix entries, and only relies on their signs.
    Also, herdability property for time varying linear systems will be the subject of future investigation.


\bibliographystyle{plain} 

 \bibliography{Refer165}

\begin{thebibliography}{10}

\bibitem{treelike1}
A.B. Adcock, B.~D. Sullivan, and M.W. Mahoney.
\newblock Tree-like structure in large social and information networks.
\newblock {\em 2013 IEEE 13th International Conference on Data Mining}, pages
  1--10, 2013.

\bibitem{treelike2}
A.B. Adcock, B.~D. Sullivan, and M.W. Mahoney.
\newblock Tree decomposition and social graphs.
\newblock {\em Internet Mathematics}, 12:315--361, 2016.

\bibitem{Altafini2013}
C.~Altafini.
\newblock Consensus problems on networks with antagonistic interactions.
\newblock {\em IEEE Trans. Aut. Contr.}, 58 (4):935--946, 2013.

\bibitem{antsaklis2007special}
P.~Antsaklis and J.~Baillieul.
\newblock Special issue on technology of networked control systems.
\newblock {\em Proc. IEEE}, 95 (1):5--8, 2007.

\bibitem{baillieul2007control}
J.~Baillieul and P.J. Antsaklis.
\newblock Control and communication challenges in networked real-time systems.
\newblock {\em Proc. IEEE}, 95(1):9--28, 2007.

\bibitem{comp_gen}
James~M. Bower and Hamid Bolouri, editors.
\newblock {\em Computational Modeling of Genetic and Biochemical Networks}.
\newblock MIT Press, 2001.

\bibitem{Bullo2020}
P.~Cisneros-Velarde and F.~Bullo.
\newblock Signed network formation games and clustering balance.
\newblock {\em Dynamic Games and Applications}, 10:783--797, 2020.

\bibitem{davis67}
J.A. Davis.
\newblock Clustering and structural balance in graphs.
\newblock {\em SAGE Social Science Collections}, 20 (2):181--187, 1957.

\bibitem{Davis}
P.J. Davis.
\newblock {\em Circulant matrices}.
\newblock J.Wiley \& Sons, New York, 1979.

\bibitem{DeGroot}
M.~H. DeGroot.
\newblock Reaching a consensus.
\newblock {\em Journal of the American Statistical Association},
  69(345):118--121, 1974.

\bibitem{DeoGraph}
N.~Deo.
\newblock {\em Graph Theory with Applications to Engineering and Computer
  Science}.
\newblock Prentice-Hall, Englewood, NJ, 1974.

\bibitem{egerstedt2012interacting}
M.~Egerstedt, S.~Martini, M.~Cao, K.~Camlibel, and A.~Bicchi.
\newblock Interacting with networks: How does structure relate to
  controllability in single-leader, consensus networks?
\newblock {\em IEEE Control Systems Magazine}, 32 (4):66--73, 2012.

\bibitem{BookFarina}
L.~Farina and S.~Rinaldi.
\newblock {\em Positive linear systems: theory and applications}.
\newblock Wiley-Interscience, Series on Pure and Applied Mathematics, New York,
  2000.

\bibitem{gupta}
S.~Gupta, S.~S. Bisht, R.~Kukreti, S.~Jain, and S.~K. Brahmachari.
\newblock Boolean network analysis of a neurotransmitter signaling pathway.
\newblock {\em J. Theoret. Biol.}, 244:463--469, 2007.

\bibitem{Heider}
F.~Heider.
\newblock Social perception and phenomenal causality.
\newblock {\em Psycological Review}, 51 (6):358--374, 1944.

\bibitem{assign_and_appraise}
E.Y. Huang, D.~Paccagnan, W.~Mei, and F.~Bullo.
\newblock Assign and appraise: achieving optimal performance in collaborative
  teams.
\newblock {\em IEEE Trans. Automatic Control}, 63(9):2898--2912, 2018.

\bibitem{Jacquez}
J.A. Jacquez.
\newblock {\em Compartmental analysis in biology and medicine}.
\newblock Elsevier, Amsterdam (NL), 1972.

\bibitem{HomophilyMei}
W.~Mei, P.~Cisneros-Velarde, G.~Chen, N.E. Friedkin, and F.~Bullo.
\newblock Dynamic social balance and convergent appraisals via homophily and
  influence mechanisms.
\newblock {\em Automatica}, 110:61--67, 2019.

\bibitem{Meng_herd}
S.~Meng, B.~She, H.~Gao, and Z.~Kan.
\newblock Leader group selection for herdability of structurally balanced
  signed networks.
\newblock {\em Conference on Decision and Control (CDC) 2020, Jeju island,
  Republic of Korea}, pages 5567--5572, 2020.

\bibitem{struct_controllability}
S.S. Mousavi, M.~Haeri, and M.~Mesbahi.
\newblock Strong structural controllability of signed networks.
\newblock {\em Conference on Decision and Control (CDC) 2019, Nice, France},
  pages 4557--4562, 2019.

\bibitem{OlfatiFaxMurray}
R.~Olfati-Saber, J.A. Fax, and R.M. Murray.
\newblock Consensus and cooperation in networked multi-agent systems.
\newblock {\em Proc. of the IEEE}, 95 (1):215--233, 2007.

\bibitem{ParlangeliNotarstefano}
G.~Parlangeli and G.~Notarstefano.
\newblock On the reachability and observability of path and cycle graphs.
\newblock {\em IEEE Trans. Automatic Control}, 57 (3):743 -- 748, 2012.

\bibitem{ge_herdcdc21}
G.~De Pasquale and M.E. Valcher.
\newblock Algebraic and graph-theoretic conditions for the herdability of
  linear time-invariant systems.
\newblock In {\em Proceedings of the 60th IEEE Conf. Decision and Control},
  Austin, Texas, USA, 2021.

\bibitem{graph_controllability}
A.~Rahmani, J.~Meng, M.~Mehran, and M.~Egerstedt.
\newblock Controllability of multi-agent systems from a graph-theoretic
  perspective.
\newblock {\em SIAM J. Control Optim.}, 48(1):162--186, 2009.

\bibitem{Ruf_Shamma}
S.F Ruf, M~Egerstedt, and J.S. Shamma.
\newblock Herdable systems over signed, directed graphs.
\newblock {\em Annual American Control Conference (ACC) 2018, Milwaukee, USA},
  pages 1807--1812, 2018.

\bibitem{Ruf_arxiv}
S.F Ruf, M~Egerstedt, and J.S. Shamma.
\newblock Herdability of linear systems bases on sign patterns and graph
  structures.
\newblock {\em arXiv:1904.08778}, 2019.

\bibitem{ZhangCaoCamlibel}
M.~Cao S.~Zhang and M.~K. Camlibel.
\newblock Upper and lower bounds for controllable subspaces of networks of
  diffusively coupled agents.
\newblock {\em IEEE Trans. Autom. Control}, 59 (3):745--750, 2014.

\bibitem{sociology1}
J.~Scott.
\newblock Social network analysis.
\newblock {\em Sociology}, 22(1):397--411, 1988.

\bibitem{She_herd}
B.~She, M.~Cai, and Z.~Kan.
\newblock Characterizing herdability of signed networks via graph walks.
\newblock {\em Conference on Decision and Control (CDC) 2019, Nice, France},
  pages 5456--5461, 2019.

\bibitem{She_automatica}
B.~She and Z.~Kan.
\newblock Characterizing controllable subspace and herdability of signed
  weighted networks via graph partition.
\newblock {\em Automatica}, 115:1--7, 2020.

\bibitem{sign_controllability}
M.~Tsatsomeros.
\newblock Sign controllability: Sign patterns that require complete
  controllability.
\newblock {\em SIAM J. Matrix Anal. Appl.}, 19(2):355--364, 1998.

\bibitem{TACconPaolo}
M.E. Valcher and P.~Santesso.
\newblock Reachability properties of single-input continuous-time positive
  switched systems.
\newblock {\em IEEE Trans. Automatic Control}, 55, no.5:1117--1130, 2010.

\bibitem{YaziciogluAbbasEger}
A.~Yazicioglu, W.~Abbas, and M.~Egerstedt.
\newblock Graph distances and network controllability.
\newblock {\em IEEE Trans. Automatic Control}, 61(12):4125--4130, 2016.

\bibitem{zhang2012network}
L.~Zhang, H.~Gao, and O.~Kaynak.
\newblock Network induced constraints in networked control systems a survey.
\newblock {\em IEEE Trans. Industrial Informatics}, 9(1):403--416, 2012.

\end{thebibliography}

\section*{Appendix}

In this Appendix we provide some technical results used within the paper.\medskip

 \begin{lemma} \label{lemmaTinput}
Consider a pair $(A,B)$, where $A\in {\mathbb R}^{n \times n}$ and $B\in {\mathbb R}^{n\times m}$. For every nonsingular matrix $T\in {\mathbb R}^{m\times m}$,
the pair $(A,B)$ is herdable if and only if the pair $(A,BT)$ is herdable.
\end{lemma}
\begin{proof}
Follows from 
${\rm Im} ({\mathcal R}(A,B))={\rm Im} ({\mathcal R}(A,BT)).$
 $\square$\end{proof}
\medskip

 The following result is elementary and its proof, that can be obtained by recursion on $k$, the number of blocks on the diagonal,  is omitted.\smallskip

\begin{lemma}\label{lemmaB}
Given a matrix   $\Phi \in \mathbb{R}^{n \times d}$, assume that there exist two permutation matrices $P_1\in {\mathbb R}^{n \times n}$ and $P_2\in {\mathbb R}^{d\times d}$ such that $P_1 \Phi P_2$ is block-partitioned as   
\begin{equation}
P_1 \Phi P_2 =	\begin{bmatrix}
\Phi_{11} & \Phi_{12}&\dots & \Phi_{1k}\\
0 & \Phi_{22} &\dots  &\Phi_{2k} \\
\vdots  & \vdots &   \ddots &\vdots \\
 0 & 0 & \dots & \Phi_{kk}
\end{bmatrix},
\label{Phiblocco}
\end{equation}
%
%
where the diagonal blocks $\Phi_{ii}$ are not necessarily square matrices.
If the image of each diagonal block $\Phi_{ii}, i\in [1,k],$ includes a strictly positive vector, then  
there exists ${\bf u} \in {\mathbb R}^k$ such that $\Phi {\bf u} \gg 0$. 
\end{lemma}
\medskip

\begin{remark} \label{remo_romolo}
As a special case, if for every $i\in [1,k]$ we can select a subset of the columns of $\Phi_{ii}$ in  \eqref{Phiblocco} such that: 1) each of them is unisigned; 2) 
for each row index $j$, at least one of these unisigned columns has the $j$-th entry which is nonzero, then 
each vector subspace ${\rm Im}(\Phi_{ii}), i\in [1,k],$ includes a strictly positive vector and hence Lemma \ref{lemmaB} ensures herdability.
\end{remark} 
\smallskip

 Proposition \ref{ablocchi_new}, below, provides a method for the dimensionality reduction of the herdability problem for matrix pairs $(A,B)$, with $A$  and $B$ 
  block-partitioned.
  \medskip

\begin{proposition} \label{ablocchi_new}
Consider a pair $(A,B)$, where $A\in {\mathbb R}^{n \times n}$ and $B\in {\mathbb R}^{n\times m}$ 
are described as in 
 as in 
\be
A =\begin{bmatrix} A_{11} & A_{21}\cr A_{21} & A_{22} \end{bmatrix},
\quad 
B= \begin{bmatrix}B_1\cr 0\end{bmatrix},
\label{blocchi2}
\ee
where  $B_1\in {\mathbb R}^{r\times m}$ is of full row rank
 and $A_{11}\in {\mathbb R}^{r\times r}$.
The pair $(A,B)$ is herdable if and only if  the pair 
$(A_{22},A_{21})$ is herdable. 
 \end{proposition}

\begin{proof}   Let ${\mathcal R}(A,B)$ be the controllability matrix of $(A,B)$ and
${\mathcal R}(A_{22},A_{21})$ the controllability matrix of $(A_{22},A_{21})$.  We preliminarily note  that, 
since $B_1$ is of full row rank, there exists a nonsingular matrix $T\in {\mathbb R}^{m\times m}$ such that $B_1 T = \begin{bmatrix} I_r & 0\end{bmatrix}$. Since the zero columns of $BT$ are irrelevant, in the following 
by making use of Lemma \ref{lemmaTinput} we   assume $r=m$ and  $B_1=I_m$.
Since 
$${\mathcal R}(A,B) = \begin{bmatrix}
I_m & \Phi_{12} \cr
0 & \Phi_{22}\end{bmatrix}$$
where
$$\!\!\begin{array}{rl}
\Phi_{22} &:=
\left[\begin{matrix} A_{21} & A_{21} A_{11} + A_{22} A_{21} \end{matrix}\right.\\
&\left.\begin{matrix} A_{21}(A_{11}^2 + A_{12} A_{21}) + A_{22}  (A_{21} A_{11} + A_{22} A_{21}) & ...
\end{matrix}\right]  \\
&= \begin{bmatrix} 0 & I_{n-m} \end{bmatrix}
\begin{bmatrix} AB & A^2 B & \dots & A^{n-1}B\end{bmatrix},
\end{array}
$$
it is immediate to see that for every ${\bf v}_1\in {\mathbb R}^m$
$${\bf v} =\begin{bmatrix} {\bf v}_1\cr {\bf v}_2\end{bmatrix} \in {\rm Im} ({\mathcal R}(A, B))
\qquad \Leftrightarrow \qquad 
{\bf v}_2 \in {\rm Im} (\Phi_{22}).$$
We now prove that
${\rm Im} (\Phi_{22}) =$ \\ ${\rm Im} \left(\begin{bmatrix} 0 & I_{n-m} \end{bmatrix}\!\!
\begin{bmatrix} AB & A^2 B & \dots & A^{n-1}B\end{bmatrix}\right) \!= {\rm Im}\left({\mathcal R}(A_{22},A_{21})\right).$
\\
To prove this result we show that for every $k\in [1, n-1]$
\be
\!\!\!\!\begin{bmatrix} 0  &  I_{n-m} \end{bmatrix}
  A^{k}B \! = \!
\begin{bmatrix} A_{21} & A_{22} A_{21} & \dots & A_{22}^{k-1} A_{21}
\end{bmatrix}  \begin{bmatrix} *\cr *\cr \vdots \cr I_{m}\end{bmatrix}
\label{perk}
\ee
where $*$ denotes a real matrix (whose value is not relevant).
We proceed by induction on $k$. If $k=1$ the result is true since
$
\begin{bmatrix} 0 & I_{n-m} \end{bmatrix}
  AB = A_{21}=
\begin{bmatrix} A_{21} 
\end{bmatrix} I_{m}.
$\\
We assume now that the result is true for $k < \bar k$ and then   show that the result is true for $k=\bar k$.
Indeed, there exists some matrix $\Xi$ such that
\begin{eqnarray*}
&&\begin{bmatrix} 0 & I_{n-m} \end{bmatrix}
  A^{\bar k}B = \begin{bmatrix} 0 & I_{n-m} \end{bmatrix} A
  A^{\bar k -1}B \\
  &=& 
  \begin{bmatrix} A_{21} & A_{22} 
\end{bmatrix} A^{\bar k -1}B = \begin{bmatrix} A_{21} & A_{22} 
\end{bmatrix} \begin{bmatrix} \Xi\cr 
\begin{bmatrix} 0 & I_{n-m} \end{bmatrix}
  A^{\bar k -1}B 
  \end{bmatrix}\\
  &=& A_{21} \Xi + A_{22} \begin{bmatrix} A_{21} & A_{22} A_{21} & \dots & A_{22}^{\bar k-2} A_{21}
\end{bmatrix} \!\! \begin{bmatrix} *\cr *\cr \vdots \cr I_{m}\end{bmatrix}\\
&=&  \begin{bmatrix} A_{21} & A_{22} A_{21} & \dots & A_{22}^{\bar k-1} A_{21}
\end{bmatrix} \begin{bmatrix} \Xi \cr *\cr \vdots \cr I_{m}\end{bmatrix}.
\end{eqnarray*}
From \eqref{perk}, applied for every $k\in [1,n-1]$, it follows that 
$$\begin{array}{l}
 \begin{bmatrix} 0 & I_{n-m} \end{bmatrix}
\begin{bmatrix} AB & A^2 B & \dots & A^{n-1}B\end{bmatrix} \\
\!\!= \!\! \begin{bmatrix} A_{21} & A_{22} A_{21} & \dots & A_{22}^{n-2} A_{21} \end{bmatrix} \!\!
\begin{bmatrix} I_m & * & \dots &*\cr
& I_m & \dots &*\cr 
&&\ddots & \vdots\cr
&&& I_m\end{bmatrix}
\end{array}
$$
and hence (by Cayley-Hamilton's theorem)
$$
\begin{array}{l}
{\rm Im} (\begin{bmatrix} 0 & I_{n-m} \end{bmatrix}
\begin{bmatrix} AB & A^2 B & \dots & A^{n-1}B\end{bmatrix}) =\cr
{\rm Im} (\begin{bmatrix} A_{21} & A_{22} A_{21} & \dots & A_{22}^{n-2} A_{21} \end{bmatrix}) 
=
 {\rm Im} ({\mathcal R}(A_{22},A_{21})).
 \end{array}
 $$
 Consequently,
the pair $(A,B)$ is herdable if and only if the pair 
$(A_{22},A_{21})$ is herdable.  $\square$\end{proof}
\medskip

\begin{remark}
The result of Proposition \ref{ablocchi_new} can be interpreted as follows: upon partitioning the state vector as ${\bf x} = [{\bf x}_1^\top,{\bf x}_2^\top ]^\top$,
conformably to the block partitioning of the matrices $A$ and $B$ in \eqref{blocchi2}, if $B_1$ is of full row rank, the corresponding subvector ${\bf x}_1$ can be 
controlled  to any point of $\mathbb{R}^m$. Therefore, the herdability of the pair $(A,B)$ is equivalent to the one of the pair $(A_{22},A_{21})$ for which the vector ${\bf x}_1$  acts as the input and the vector ${\bf x}_2$ as the state.
\end{remark}

%
%

\begin{lemma} \label{vandermonde}
 Given a matrix pair $(\Lambda, \Gamma)$, with $\Lambda = {\rm diag}\{\lambda_1, \lambda_2, \dots, \lambda_n\} \in {\mathbb R}^{n\times n}$, 
  and $\Gamma\in {\mathbb R}^{n}$ devoid of zero entries, the pair is herdable if and only 
  $\lambda_i=\lambda_j$ 
  implies $[\Gamma]_i\cdot [\Gamma]_j > 0$, namely the $i$-th and the $j$-th entries of $\Gamma$ have the same sign.
  \end{lemma}
  
  \begin{proof} We first prove that if the pair $(\Lambda, \Gamma)$ is herdable, then $\lambda_i=\lambda_j$ 
  implies $[\Gamma]_i\cdot [\Gamma]_j > 0$.
  Suppose, by contradiction, that $\lambda_i=\lambda_j=: \lambda$ and  $[\Gamma]_i\cdot [\Gamma]_j < 0$.
  Then it is easy to see that $i\ne j$ and the vector
  ${\bf w}^\top := [\Gamma]_j {\bf e}_i^\top - [\Gamma]_i {\bf e}_j^\top$ 
  satisfies 
  ${\bf w}^\top \Lambda = \lambda {\bf w}^\top,$
  namely ${\bf w}$ is a (left) eigenvector of $\Lambda$ corresponding to $\lambda$,
  and ${\bf w}^\top \Gamma=0$. Consequently, it is immediate to prove that  ${\bf w}$ is orthogonal to ${\rm Im}(\mathcal{R}(\Lambda, \Gamma))$, i.e. 
  ${\bf w}^\top {\mathcal R}(\Lambda, \Gamma)={\bf 0}_n^\top.$
  Since ${\bf w}^T$  is unisigned (since $[\Gamma]_j$ and $- [\Gamma]_i$ have the same sign), it is impossible that there exists a strictly positive vector
  ${\bf v}\in {\rm Im} ({\mathcal R}(\Lambda, \Gamma))$, since this would imply ${\bf w}^\top {\bf v} \ne 0$. Therefore
  the pair $(\Lambda, \Gamma)$ cannot be herdable.
  \smallskip
  
  We now prove that if $\lambda_i=\lambda_j$ 
  implies $[\Gamma]_i\cdot [\Gamma]_j > 0$, then the pair $(\Lambda, \Gamma)$ is herdable.
  \\
 If all entries of $\Gamma$ have the same sign, the pair  $(\Lambda, \Gamma)$ is trivially herdable. So, suppose this is not the case.
  It entails no loss of generality to first permute the entries of $\Gamma$ (and accordingly the rows and columns of $\Lambda$)
  so that the first ones are positive and the last ones are negative. Then we can permute the entries in such a way that
 the identical diagonal entries of $\Lambda$ are consecutive.
 This implies that, under the previous assumptions,  $\Lambda$ and $\Gamma$ take the following form
\be
\Lambda =
 \begin{bmatrix}
 \Lambda_1 & & &\vline& & & \cr
 &\ddots & &\vline&&&\cr
 & & \Lambda_p&\vline &&&\cr
 \hline
 & & & \vline &\Lambda_{p+1} & & \cr
  & & & \vline& &\ddots& \cr
 & & & \vline& && \Lambda_s 
 \end{bmatrix} \ \ \ \
\Gamma = \begin{bmatrix}
 {\bm \gamma}_1\cr \vdots\cr {\bm \gamma}_p\cr \hline
 {\bm \gamma}_{p+1}\cr \vdots\cr {\bm \gamma}_s
\end{bmatrix}
\label{lambdagamma}
\ee

\noindent   where each $\Lambda_i$ is a scalar matrix of size say $n_i$, namely $\Lambda_i=\bar \lambda_i I_{n_i}$ with $\bar \lambda_i\in \{\lambda_1, \dots, \lambda_n\}$,
while ${\bm \gamma}_i\in {\mathbb R}^{n_i}$ is a strictly positive vector for every $i\in [1,p]$ and a strictly negative vector for every $i\in [p+1,s]$. 
Moreover, by the assumption that $\lambda_i=\lambda_j$ 
  implies $[\Gamma]_i\cdot [\Gamma]_j > 0$ 
   we can claim that $\bar \lambda_h\ne \bar \lambda_k$ for $h\ne k$.
It is immediate to see that the controllability matrix of the pair $(\Lambda, \Gamma)$ factorizes as in \eqref{factor}.
\be
 {\mathcal R}(\Lambda, \Gamma)\!=\! 
\begin{bmatrix}
{\bm \gamma}_1 & & &\!\!\vline\!\!& & & \cr
 &\!\! \ddots \!\!  & &\!\!\vline\!\!&&&\cr
 & & {\bm \gamma}_p&\!\!\vline\!\!&&&\cr
 \hline
 & & &\!\!\vline\!\!& {\bm \gamma}_{p+1} & & \cr
  & & &\!\!\vline\!\!&\!\!  &\!\!  \ddots\!\! & \cr
 & & &\!\!\vline\!\!& &&  {\bm \gamma}_s 
 \end{bmatrix}
 \begin{bmatrix}
 1\!\!  & \bar \lambda_1\!\! & \dots &\!\! \bar\lambda_1^{n-1}\cr
\vdots\!\! &\vdots\!\! &\dots &\!\!\vdots\cr
 1 \!\! & \bar\lambda_p\!\! & \dots & \!\!\bar\lambda_p^{n-1}\cr
  1\!\!  & \bar\lambda_{p+1}\!\! & \dots &\!\! \bar\lambda_{p+1}^{n-1}\cr
\vdots\!\! &\vdots\!\! &\dots &\!\!\vdots\cr
 1\!\!  & \bar\lambda_s\!\! & \dots &\!\! \bar\lambda_s^{n-1}
 \end{bmatrix}\!\!
\label{factor}
\ee


\noindent Since $\bar\lambda_1, \dots, \bar\lambda_s$ are all distinct, the Vandermonde matrix on the right of \eqref{factor} is of full row rank. This ensures that
$${\rm Im}  ({\mathcal R}(\Lambda, \Gamma)) =
{\rm Im}\left(\begin{bmatrix}
{\bm \gamma}_1 & & &\vline& & & \cr
 &\ddots & &\vline&&&\cr
 & & {\bm \gamma}_p&\vline &&&\cr
 \hline
 & & & \vline & {\bm \gamma}_{p+1} & & \cr
  & & & \vline& &\ddots& \cr
 & & & \vline& && {\bm \gamma}_s 
 \end{bmatrix}\right)$$
 and since all columns of this latter matrix are unisigned, it is immediate to see
 that there exists a strictly positive vector in its image, and hence
the pair $(\Lambda, \Gamma)$ is herdable.  $\square$\end{proof}

\end{document}